\documentclass[a4paper,12pt]{elsarticle}
\journal{Journal of Graph Theory}
\usepackage{algorithm}
\usepackage{algpseudocode}
\algnewcommand\algorithmicinput{\textbf{Input:}}
\algnewcommand\algorithmicoutput{\textbf{Output:}}
\algnewcommand\Input{\item[\algorithmicinput]}
\algnewcommand\Output{\item[\algorithmicoutput]}

\usepackage{amssymb,amsmath, amsthm,enumerate}
\usepackage{graphicx}
\usepackage{tikz} 
\usepackage{soul}
\DeclareMathOperator{\lexlabel}{label}
\DeclareMathOperator{\lcm}{lcm}
\DeclareMathOperator{\LexCycle}{LexCycle}
\DeclareMathOperator{\LexBFS}{LexBFS}
\DeclareMathOperator{\diff}{diff}

\usetikzlibrary{calc,shapes.multipart,chains,arrows,positioning}

\newtheorem{theorem}{Theorem}[section]
\newtheorem{proposition}[theorem]{Proposition}
\newtheorem{lemma}[theorem]{Lemma}

\newtheorem{corollary}[theorem]{Corollary}

\newtheorem{conjecture}[theorem]{Conjecture}
\newtheorem{definition}[theorem]{Definition}

\newtheorem{question}[theorem]{Question}
\newtheorem{property}[theorem]{Property}

\usepackage[textwidth = 3cm,textsize=small
    ]{todonotes}

\begin{document}
	\begin{frontmatter}

\title{A New Graph Parameter To Measure Linearity \tnoteref{t1}}
\tnotetext[t1]{Supported by the ANR-France project HOSIGRA (ANR-17-CE40-0022) and NSERC.}

\author[IRIF]{Pierre Charbit} 
\author[IRIF]{Michel Habib}
\author[Toronto]{Lalla Mouatadid} 
\author[IRIF]{Reza Naserasr}
\address[IRIF]{IRIF, CNRS  \&  Université Paris Cité, Paris, France. email: charbit, habib, reza, @irif.fr}
\address[Toronto]{Department of Computer Science, University of Toronto, Toronto, Ontario, Canada, email: lalla@cs.toronto.edu}

\begin{abstract}
Consider a sequence of $\LexBFS$ vertex orderings $\sigma_1, \sigma_2, \ldots$ where each ordering $\sigma_i$ is used to break ties for $\sigma_{i+1}$.   
Since the total number of vertex orderings of a finite graph is finite, this sequence must end in a cycle of vertex orderings. The possible length of this cycle is the main subject of this work. 
Intuitively, we prove for graphs with a known notion of linearity (e.g., interval graphs with their interval representation on the real line), this cycle cannot be too big, no matter which vertex ordering we start with.
More precisely, it was conjectured in~\cite{DH15} that for cocomparability graphs, the size of this cycle is always 2, independent of the starting order. Furthermore~\cite{Juraj} asked whether for arbitrary graphs, the size of such a cycle is always bounded by the asteroidal number of the graph. 
In this work, while we answer this latter question negatively, we provide support for the conjecture on cocomparability graphs by proving it for the subclass of domino-free cocomparability graphs. This subclass contains cographs, proper interval, interval, and cobipartite graphs. 
We also provide simpler independent proofs for each of these cases which lead to stronger results on this subclasses. 
\end{abstract}
\end{frontmatter}
\section{Introduction} 
\label{intro}
A \emph{graph search} or a \emph{graph traversal} is a mechanism to visit the vertices of a graph. 
Depth-First Search (DFS) and Breadth-First Search (BFS) are two classical and well studied examples of such traversals. 
If a graph search visits every vertex exactly once, then it produces a total ordering of the vertices of the graph corresponding to the order in which they are visited. The different searches can be therefore analyzed through the properties of the vertex orderings they produce.

Graph searches are often described by a criterion deciding, given an initial segment of the ordering, which vertex can be placed next. For instance, if we start a BFS at a vertex $v$, then all the neighbours of $v$ must be visited before the non-neighbours of $v$. Also, most of the times, there are so called tied vertices, i.e. several vertices that are simultaneously eligible to be placed next, and thus an arbitrary choice can be made. For example in BFS, once the root is chosen, the ordering in which its neighbours are visited can be arbitrary. 

Given a graph search, such as BFS, one can thus define a more precise graph search simply by defining tie-breaking rules, and this has proved to be a powerful technique to understand and analyze the structure of certain graph classes.
This line of work originally started in 1976 by Rose, Tarjan, and Lueker, when they introduced the lexicographic variant of BFS in~\cite{RoseTL76}, known as \emph{lexicographic breadth first search}, or $\LexBFS$ for short. 
One of the first uses of this graph search was the simplest linear time algorithm to recognize chordal graphs~\cite{RoseTL76}. 
Since then, $\LexBFS$ has led to a number of simple, efficient, and elegant algorithms on various graph classes~\cite{CorneilOS09,DH15,KratschMMS03}.

One way to break \emph{all} ties while constructing an ordering $\tau$ consists in using another ordering $\sigma$ : if there is a tie between two vertices $x$ and $y$, one shall pick the one that is the ``greatest" in $\sigma$. 
This was introduced by Simon in~\cite{simon1991} for LexBFS and is known as the  $^+$ rule. 
Given an order $\sigma$ on the vertices of $G$, Simon defines $\LexBFS^+(G,\sigma)$ as the (unique) LexBFS ordering of the vertices of $G$ obtained by breaking ties by always picking the right most vertex with respect to $\sigma$ (for instance, $\LexBFS^+(G, \sigma)$ starts with the last vertex of $\sigma$).
Now given an initial ordering $\sigma_0$ on the vertices of $G$, one can thus define a sequence $\sigma_0, \sigma_1, \sigma_2, \ldots$ of orderings on $V$ by setting $\sigma_{i}=\LexBFS^+(G, \sigma_{i-1})$. 
This technique is known as a \emph{multisweep algorithm} and has been used to introduce fast recognition algorithms for graph classes such as proper interval, interval, and cocomparability graphs~\cite{Corneil04,CorneilOS09,DH15}. 
The idea here is to prove some kind of convergence to say that this process will eventually yield some vertex ordering with strong structural properties. 
This technique is of course especially relevant for the study of graph classes which are defined, or characterized, by the existence of certain types of vertex orderings. 
For instance unit interval graphs are defined as intersection graphs of interval of length $1$ of the real line, but it is a classical theorem that they are exactly the graphs whose vertex set can be ordered such that  for any three vertices  $a,b,c$ with $a\prec b\prec c$, $ac\in E$ implies that  $ab\in E$ and $bc \in E$. 
In \cite{Corneil04} a very simple certifying recognition algorithm based on $\LexBFS^+$ is given : starting from any ordering, $3$ sweeps must provide such an order (which is easy to check) if the input graph is unit interval.

Evidently, as the number of distinct vertex orderings of a finite graph is finite, no matter which ordering $\sigma_0$ we start with, this sequence $\{\sigma_i\}_{i\ge1}$ of $\LexBFS^+$ orderings will eventually cycle. That is, for some $i$ and $k$, $\sigma_{i+k}=\sigma_i$.
For general graphs this observation raises two interesting questions :
\begin{enumerate}[(i)]
    \item Among all possible choices of $\sigma_0$ as a start ordering, how long does it take to reach a cycle?

    \item How large can this cycle be?
\end{enumerate}

This paper is concerned with these questions for the class of cocomparability graphs, a superclass of interval graphs characterized by the existence of a so called \emph{cocomparability ordering} : for any three vertices  $a,b,c$ with $a\prec b\prec c$, $ac\in E$ implies that  $ab\in E$ or $bc \in E$ (such an order is a transitive order - i.e. a linear extension of a transitive orientation - of the complement graph, hence the name of the class). 

One important reason for restricting our attention to cocomparability graphs is because Dusart and Habib proved the following theorem.
\begin{theorem}\label{thm:DusartHabib}\cite{DH15}
	If $G$ is a cocomparability graph on $n$ vertices, and $\sigma_0$ an arbitrary ordering of $V(G)$, define a sequence $\{\sigma_i \}_{i\ge 1}$ of $\LexBFS^+$ orderings of $G$ as $\sigma_{\color{black}i}=\LexBFS(G, \sigma_{i-1}) $. Then $\sigma_n$ is a cocomparability ordering of the vertices of $G$.
\end{theorem}

While this theorem guarantees for cocomparability graphs that a multisweep process will reach a cocomparability ordering in at most $n$ iterations, we don't know in general any non-trivial bound on when the cycle will be reached. 
For some subclasses of cocomparability graphs, we prove such bonds in this paper.

Regarding the second question above (ii), and again restricted to the class of cocomparability graphs, Dusart and Habib \cite{DH15} have conjectured that, no matter which initial ordering we start with, the length of the cycle is at most $2$ (a cycle of length $1$ being in fact impossible except for the one vertex graph, since the last vertex of an order is always the first vertex of the next order).
\begin{conjecture}\label{conj:DusartHabib}
    Given a cocomparability graph $G$, an arbitrary ordering $\sigma_0$ of $V(G)$, and a sequence $\{\sigma_i\}_{i\ge1}$ of $\LexBFS^+$ orderings of $G$ where $\sigma_i$ is used to break ties for $\sigma_{i+1}$, for $i$ {\color{black}sufficiently large}, we have $\sigma_{i}=\sigma_{i+2}$.	
\end{conjecture}
Observing that cocomparability graphs are asteroidal triple-free, and thus have asteroidal number two, Stacho asked if the length of all such cycles is bounded by the asteroidal number of the graph~\cite{Juraj} . 

In this work, we first answer Stacho's question negatively. 
Then, we provide strong support for the conjecture of Dusart and Habib by proving it for cocomparability graphs that do not contain a particular $6$ vertex graph (called domino) as an induced subgraph. 
While this subclass of cocomparability graphs contains proper interval graphs, interval graphs, cographs and cobipartite graphs, we additionally give for each of these cases an independent proof which provides stronger results, and sheds light into structural properties of these graph classes. 

The structure of the paper is as follows: we finish this introduction section by giving basic definitions and fixing our notations. 
In Section~\ref{sec:multisweep} we give all the necessary background to understand $\LexBFS$ properties and its use in multisweep algorithms. 
We also introduce, define, and discuss $\LexCycle(G)$, the main invariant studied in our paper. 
In particular, we give a construction that gives an answer to the question of Stacho mentioned earlier. 
In Section~\ref{sec:vocs}, we expose various results related to vertex ordering characterizations of the classes of graphs and the graph searches studied in the paper. 
Section~\ref{sec:classes} contains our main results mentioned in the previous paragraph about Conjecture \ref{conj:DusartHabib} in the subclass of domino-free cocomparability graphs. 
Finally in Section~\ref{conclusion} we present further ideas,  and research directions. 

\subsection{Notations}
A graph $G$ is a pair $(V, E)$ where $V$ is a finite set whose elements are called vertices, and $E$ is a set of unordered pairs of $V$ called edges.
We sometimes write $V(G)$ and $E(G)$ to denote the vertices and the edges of a graph $G$. If no ambiguity occurs, we will always use the letters $n$ and $m$ to denote respectively the number of vertices and edges of a graph $G$.
Given a pair of adjacent vertices $u$ and $v$, we write $uv$ to denote the edge in $E$ with endpoints $u$ and $v$. 
We denote by $N(v) = \{u : uv \in E\}$ the open neighbourhood of vertex $v$, and $N[v] = N(v) \cup \{v\}$ the closed neighbourhood of $v$. 
We write $G[V']$ to denote the \emph{induced subgraph} $(V',E')$ of $G=(V,E)$ on the subset $V'$ of $V$, where for every pair $u,v \in V', uv \in E'$ if and only if $uv \in E$. 
{\color{black}A graph class $\cal G$ is said to be \emph{hereditary} if it is closed under induced subgraphs.} 
The complement of a graph $G=(V,E)$ is the graph $\overline{G}(V, \overline{E})$ where $uv \in \overline{E}$ if and only if $uv \notin E$.  
A \emph{private neighbour} of a vertex $u$ with respect to a vertex $v$ is a third vertex $w$ that is adjacent to $u$ but not $v$: $uw \in E, vw \notin E$.

A set $S \subseteq V$ is an independent set if for all $a,b \in S, ab \notin E$, and is a clique set if for all $a,b \in S, ab \in E$. 
Given a pair of vertices $u$ and $v$, the distance between $u$ and $v$, denoted $d(u,v)$, is the length of a shortest $u,v$ path. A \emph{diametral} path of a graph is a shortest $u,v$ path where $u$ and $v$ are at the maximum distance among all pairs of vertices. 
A \emph{dominating} path in a graph is a path where all the vertices of the graph are either on the path or have a neighbour on the path.
A triple of independent vertices $u,v,w$ forms an \emph{asteroidal triple} (AT) if every pair of the triple remains connected when the third vertex and its closed neighbourhood are removed from the graph. In general, a set $A$ of vertices of $G$ forms an \emph{asteroidal set} if for 
each vertex $a \in A$, the set $A\backslash\{a\}$ is contained in one connected component of 
$G[V\backslash N[a]]$. The maximum cardinality of an asteroidal set of $G$, denoted $an(G)$, is called the \emph{asteroidal number} of $G$. 
A graph is \emph{AT-free} if it does not contain an asteroidal triple. 
The class of AT-free graphs contains cocomparability graphs.
A \emph{domino} (Fig.~\ref{fig:domino}) is the induced graph $G=(V = \{a,b,c,d,e,f\}, E = \{ab, ac, bd, cd, ce, df, ef\})$.
\begin{figure}[H]
	\centering
	\begin{tikzpicture}[scale=.3]
	    \node[circle, draw, fill=black!100, inner sep=1pt, minimum width=4pt,label=above:$a$] (a) at (0,0) {};
	    \node[circle, draw, fill=black!100, inner sep=1pt, minimum width=4pt,label=below:$b$] (b) at (0,-4) {};
	    \node[circle, draw, fill=black!100, inner sep=1pt, minimum width=4pt,label=above:$c$] (c) at (4,0) {};
	    \node[circle, draw, fill=black!100, inner sep=1pt, minimum width=4pt,label=below:$d$] (d) at (4,-4) {};
	    \node[circle, draw, fill=black!100, inner sep=1pt, minimum width=4pt,label=above:$e$] (e) at (8,0) {};
	    \node[circle, draw, fill=black!100, inner sep=1pt, minimum width=4pt,label=below:$f$] (f) at (8,-4) {};
	    \foreach \from/\to in {a/b,a/c,b/d,c/d,c/e,d/f,e/f} \draw (\from) -- (\to);
	\end{tikzpicture}
	\caption{Domino}\label{fig:domino}
\end{figure}

Let $[k]$ denote the set of integers $1$ to $k$. 
Given a graph $G=(V, E)$, an \emph{ordering} $\sigma$ of $G$ is a bijection $\sigma: V \leftrightarrow 
[n]$. 
For $v \in V$, $\sigma(v)$ refers to the position of $v$ in $\sigma$. 
For a pair $u, v$ of vertices  we write $u \prec_{\sigma} v$ if and only if $\sigma(u) < 
\sigma(v)$; we also say that $u$ (resp. $v$) \emph{is to the left of} (resp. \emph{right of}) $v$ (resp. $u$). 
We write $\{\sigma_i\}_{i\ge 1}$ to denote a sequence of orderings $\sigma_1, \sigma_2, \ldots$. 
We also write $\sigma_{i>1}$ to denote an ordering $\sigma_i$ where $i > 1$. 

Given a sequence  of orderings $\{\sigma_i\}_{i\ge 1}$ of a graph $G$, and an edge $ab \in E$, we write $a \prec_i b$ if $a \prec_{\sigma_i} b$, and $a \prec_{i,j} b$ if $a \prec_i b$ and $a \prec_j b$. 
Given an ordering $\sigma = v_1, v_2, \ldots, v_n$ of $G$, we write $\sigma^d$ to denote the \emph{dual} (also called \emph{reverse}) ordering of $\sigma$; that is $\sigma^d = v_n, v_{n-1}, \ldots, v_2, v_1$. 
For an ordering $\sigma = v_1, v_2, \ldots, v_n$, the interval $\sigma[v_s,\ldots,v_t]$ denotes the ordering of $\sigma$ restricted to the vertices $\{v_s,v_{s+1},\ldots,v_t\}$ as numbered by $\sigma$. Similarly, if $S \subseteq V$, and $\sigma$ an ordering of $V$, we write $\sigma[S]$ to denote the ordering of $\sigma$ restricted to the vertices of $S$. 

\section{LexBFS, multisweep Algorithms and LexCycle}
\label{sec:multisweep}
A \emph{multisweep algorithm} is an algorithm that computes a sequence of orderings where each ordering $\sigma_{i}$ uses the previous ordering $\sigma_{i-1}$ to break ties using some predefined tie-breaking rules. 
We focus on one specific tie-breaking rule: \textbf{the $^+$ rule}, formally defined as follows: Given a graph $G=(V,E)$, an ordering $\sigma$ of $G$, and a graph search $S$ (such as LexBFS), $S^+(G, \sigma)$ is a new ordering $\tau$ of $G$ that uses $\sigma$ to break any remaining ties from the $S$ search. 
In particular, given a set $T$ of tied vertices, the $^+$ rule chooses the vertex in $T$ that is rightmost in $\sigma$. 
We sometimes write $\tau = S^+(\sigma)$ instead of $\tau = S^+(G, \sigma)$ if there is no ambiguity on the graph considered.

In this work, we focus on LexBFS based multisweep algorithms. 
LexBFS is a variant of BFS that assigns lexicographic labels to vertices, and breaks ties between them by choosing vertices with lexicographically highest labels. 
The labels are words over the alphabet {\color{black}$\{1,...,n\}$}. 
 We denote by $\lexlabel(v)$ the label of a vertex $v$.
By convention $\epsilon$ denotes the  empty word. 
LexBFS was initially introduced by Rose, Tarjan, and Lueker to recognize chordal graphs~\cite{RoseTL76}. 
We present LexBFS in Algorithm~\ref{lexbfsalg} below. 
The operation \emph{append$(n-i)$} in Algorithm \ref{lexbfsalg}, puts the letter $n-i$ at the end of the word. 
\begin{algorithm}[H]
	\caption{LexBFS}\label{lexbfsalg}
	\begin{algorithmic}[1]
		\Input A graph $G=(V,E)$ and a start vertex $s$
		\Output An ordering $\sigma$ of $V$
		\State assign the label $\epsilon$ to all vertices, and $\lexlabel(s) \leftarrow {\color{black}\{n\}}$
		\For{$i \leftarrow 1$ to $n$}
		\State pick an unnumbered vertex $v$ with lexicographically largest label
		\State $\sigma(v) \leftarrow i$ \Comment{$v$ is assigned the number $i$}
		\ForAll{unnumbered vertex $w$ adjacent to $v$}
		
		\State append$(n-i)$ to $\lexlabel(w)$
		\EndFor
		\EndFor
	\end{algorithmic}
\end{algorithm}

Starting from an ordering $\sigma_0$ of $G$, a multisweep $\LexBFS^+$ process consists of computing the following sequence:
$\sigma_{i+1}=\LexBFS^+(G, \sigma_i)$. 
Since $G$ has a finite number of LexBFS orderings, such a sequence must get into a finite cycle of vertex orderings. This leads to the definition below, notice that there is no assumption on the starting vertex ordering $\sigma_0$. 
\begin{definition}[LexCycle]\label{definitionOfLexCycle}
	For a graph $G=(V,E)$, let $\LexCycle(G)$ be the \textbf{maximum} length of a cycle of vertex orderings obtained via a sequence of LexBFS$^+$ sweeps.
\end{definition}

Note that contrary to other classical invariants, it is not at all clear whether this should be a monotone function for the induced subgraph relation. 
The following question is still open, even for cocomparability graphs.

\begin{question}
If $H$ is an induced subgraph of $G$, is it true that $\LexCycle(H)$ is at most $\LexCycle(G)$?
\end{question}

Another viewpoint on $\LexCycle(G)$ is obtained by constructing a directed graph $G_{lex}$ whose vertices are all $\LexBFS$ orderings of $G$, and with an arc from $\sigma$ to $\tau$ if $\LexBFS^+(G, \sigma) = \tau$. 
The digraph $G_{lex}$ is a functional digraph : every vertex has an out-degree of exactly one, and therefore every connected component of $G_{lex}$ is a circuit on which are planted some directed trees. 
For instance, if $K$ is a clique, $K_{lex}$ is just the union of directed circuits of size two joining one permutation to its reverse. 
$\LexCycle(G)$ is then just the maximum size of a directed circuit in $G_{lex}$, and we do not know of any example of a graph with two distinct cycle lengths.\\

In this work, we study the first properties of this new graph invariant, $\LexCycle$. 
Due to the nature of the $^+$ rule, $\LexCycle(G) \geq 2$ as soon as $G$ contains more than one vertex (the last vertex of an order is the first vertex of the next one).
Obviously $\LexCycle(G) \leq n!$, and more precisely $\LexCycle(G)$ is bounded by the number of LexBFS orderings of $G$.
We introduce a construction, \emph{Starjoin}, below which suggests (but does not yet prove) that solely based on the number of vertices and without the use of the structural constraints on the graph, we cannot bound $\LexCycle(G)$ by a polynomial on $n$. 
This construction will allow us, at the end of this section, to answer the question of Stacho~\cite{Juraj} mentioned in the introduction, which asks if  $\LexCycle(G) \leq an(G)$ for any graph. 

We start by constructing some graphs with $\LexCycle\ge 3$ :
\begin{itemize}
    \item $G_3$ is the graph represented on Figure~\ref{G3}. It satisfies $\LexCycle(G_3) \geq 3 = an(G_3)$, as shown by the multisweep starting with $\sigma_1 = x, b, a, c, e, f, d, z, y$. 

    \item $G_4$ is the graph represented on Figure~\ref{G4}. It satisfies $\LexCycle(G_4) \geq 4 = an(G_4)$, as shown by the multisweep starting with  $\mu_1 = \LexBFS(G) = x_4, z_4, y_1, y_3, y_4, y_2, z_2, z_1, z_3, x_2, x_3, x_1$. 
\end{itemize}
\begin{figure}[H]
	\begin{minipage}{.4\textwidth}
		\centering
		\begin{tikzpicture}[auto=left,every node/.style={circle,inner sep=1pt, minimum width=4pt,draw,fill=black!100}]
		\node[label=below:$f$] (f) at (0.5, 1.5) {};
		\node[label=below:$y$] (y) at (2, 1.5) {};
		\node[label=below:$x$] (x) at (-1.5-0.25, 0.5) {};
		\node[label=below:$a$] (a) at (-0.5+0.25, 0.5) {};
		\node[label=below:$e$] (e) at (0.5+0.75, 0.5) {};
		\node[label=above:$b$] (b) at (-1, -0.5) {};
		\node[label=above:$c$] (c) at (0+0.5, -0.5) {};
		\node[label=above:$d$] (d) at (2, -0.5) {};
		\node[label=above:$z$] (z) at (1.25, -1.5) {};
		\foreach \from/\to in {f/y, f/a, f/e, y/e, a/e, a/x, a/b, a/c, e/c, x/b, b/c, c/d, e/d, c/d, c/z, d/z}
		\draw (\from) -- (\to);
		\end{tikzpicture}
	\end{minipage}
	\begin{minipage}{0pt}
		\centering
		\begin{align*}
		&\text{LexBFS}(G) &= \sigma_1 &=  x, b, a, c, e, f, d, z, y\\
		&\text{LexBFS}^+(\sigma_1) &= \sigma_2 &= y, f, e, a, c, d, b, x, z\\
		&\text{LexBFS}^+(\sigma_2) &= \sigma_3 &= z, d, c, e, a, b, f, y, x\\
		&\text{LexBFS}^+(\sigma_3) &= \sigma_1 &= x, b, a, c, e, f, d, z, y
		\end{align*}
	\end{minipage}
	\caption{Example of a graph with $\LexCycle(G_3) \ge 3$ where the 3-cycle consists of $C_3=[\sigma_1, \sigma_2, \sigma_3]$.}\label{G3}
\end{figure}
\begin{figure}[H]
	\begin{minipage}{.4\textwidth}
		\centering
		\begin{tikzpicture}[auto=left,every node/.style={circle,inner sep=1pt, minimum width=4pt,draw,fill=black!100}]
		\node[label=above:$z_1$] (f) at (0.5, 1.5) {};
		\node[label=above:$x_1$] (y) at (2, 1.5) {};
		\node[label=above:$x_4$] (x) at (-1.5-0.25, 1) {};
		\node[label=above:$y_1$] (a) at (-0.5+0.25, 0.5) {};
		\node[label=above:$y_2$] (e) at (0.5+0.75, 0.5) {};
		\node[label=left:$y_4$] (b) at (-0.5+0.25, -0.5) {};
		\node[label=below:$y_3$] (c) at (0.75+0.5, -0.5) {};
		\node[label=above:$z_2$] (d) at (2.5, 0) {};
		\node[label=right:$x_2$] (z) at (2, -1.25) {};
		\node[label=below:$x_3$] (u) at (-1, -1.50) {};
		\node[label=below:$z_3$] (v) at (0.5, -1.25) {};
		\node[label=left:$z_4$] (w) at (-1.5, 0) {};
		\foreach \from/\to in {f/y, f/a, f/e, y/e, a/e, a/x, a/b, a/c, e/c, x/w, b/c, c/d, e/d, c/d, c/z, d/z, e/b, u/v, v/c, v/b, u/b, w/b, w/a, f/b,a/d,e/v,w/c}
		\draw (\from) -- (\to);
		\end{tikzpicture}
	\end{minipage}
	\begin{minipage}{0pt}
		\centering
		\begin{align*}
		& & \mu_1 &= x_4, z_4, y_1, y_3, y_4, y_2, z_2, z_1, z_3, x_2, x_3, x_1\\
		&\mu_1^+ &= \mu_2&= x_1, z_1, y_2, y_4, y_1, y_3, z_3, z_2, z_4, x_3, x_4, x_2\\
		&\mu_2^+ &= \mu_3 &= x_2, z_2, y_3, y_1, y_2, y_4, z_4, z_3, z_1, x_4, x_1, x_3\\
		&\mu_3^+ &=\mu_4&=  x_3, z_3, y_4, y_2, y_3, y_1, z_1, z_4, z_2, x_1, x_2, x_4 \\
		&\mu_4^+ &= \mu_1&= x_4, z_4, y_1, y_3, y_4, y_2, z_2, z_1, z_3, x_2, x_3, x_1\\
		\end{align*} 
	\end{minipage}
	\caption{Example of a graph with $\LexCycle(G_4) \ge 4$ where the 4-cycle consists of $C_4=[\mu_1, \mu_2, \mu_3, \mu_4]$.}\label{G4}
\end{figure}

We now show how one can construct graphs with $\LexCycle(G) > an(G)$. 
Consider the following graph operation that we call \emph{Starjoin}.
\begin{definition}[Starjoin]
	For a family of vertex disjoint connected graphs $\{G_i\}_{1\leq i \leq k}$, we define $H=Starjoin(G_1, \dots G_k)$ as follows:
	For $i \in [k]$, add a universal vertex $g_i$ to $G_i$, then add a root vertex $r$ adjacent to all $g_i$'s. 
\end{definition}

\begin{proposition}\label{starjoin}
	Let $G_i$ be a graph with a cycle $C_i$ in a sequence of $\LexBFS^+$ orderings of  $G_i$ and let $H=Starjoin(G_1, \dots G_k)$. We have 
	\begin{itemize}
		\item $an(H)=\max \{k, an(G_1), an(G_2), \ldots, an(G_k)  \}$
		\item$\LexCycle(H) \geq \lcm_{1 \leq i \leq k}\{ |C_i|\}$, where $\lcm$ stands for the least common multiple.
	\end{itemize}
\end{proposition}
\begin{proof}
	Notice first that selecting one vertex per $G_i$ would create a $k$-asteroidal set. 
	Since every $g_i$ vertex is universal to $G_i$, we can easily see that every asteroidal set of $H$ is either restricted to one $G_i$, or it contains at most one vertex per $G_i$. This yields the first formula. 
	
	For the second property, we notice first that a cycle of $\LexBFS^+$ orderings is completely determined by its initial $\LexBFS$ ordering, since all ties are resolved using the $^+$ rule.  
	For $1 \le i \le k$, let $\sigma_1^i$ denote the first $\LexBFS^+$ ordering on $C_i$, the cycle in a sequence of $\LexBFS^+$ orderings of $G_i$. 
	
	Consider the following $\LexBFS$ ordering of $H$: $\sigma_1^H= r, g_1, \dots g_k \sigma_1^1,   \dots \sigma_1^k$.
	{\color{black} Consider the cycle of $\LexBFS^+$ orderings that will result after running a sequence of $\LexBFS^+$, starting with $\sigma_1^H$ as its first ordering.}
	Notice that in any $\LexBFS^+$ ordering {\color{black}in this cycle,} the vertices of $G_i$ are consecutive, with the exception of $g_i$ that can appear in between $G_i$'s vertices. 
	Furthermore $\sigma_j^H[G_i]= \LexBFS^+(G_i, \sigma_{j-1}^i)$.
	Therefore if we take $\sigma_1^i$ as the first $\LexBFS^+$ ordering of $C_i$, then the length of the cycle generated by $\sigma_1^H$ is necessarily a multiple of $|C_i|$. 
\end{proof}

We are now ready to answer Stacho's conjecture negatively.

\begin{corollary}
	There exists a graph $G$ satisfying $LexCycle(G) > an(G)$.
\end{corollary}
\begin{proof}
	To see this, consider $H=Starjoin(G_3, G_4)$ constructed using the graphs in Figures~\ref{G3} and~\ref{G4}. By Proposition~\ref{starjoin},  $an(H)={4}$ and $\LexCycle(H)\geq 12$. 
\end{proof}

A natural question to raise here is whether $\LexCycle$ can be bounded by some function of the asteroidal number. 
In order to disprove this fact, it would be enough by Proposition~\ref{starjoin} to generalize the constructions of $G_3$ and $G_4$ to graphs with bounded asteroidal number but arbitrarily large prime $\LexCycle$ values. 
We do not have such a generalization yet. 
\section{Vertex Ordering Characterizations of Classes and Searches}
\label{sec:vocs}
Given a graph class $\mathcal{G}$, a \emph{vertex ordering characterization} (or VOC) of $\mathcal{G}$ is a characterization of a graph class given by the existence of a total ordering on the vertices with specific properties.  
VOCs have led to a number of efficient algorithms, and are often the basis of various graph recognition algorithms, see for instance~\cite{RoseTL76, corneil2013ldfs, CorneilOS09, kohler2016linear, habib2017maximum}. 
In this section, we describe some of these VOCs for the graph classes for which we will prove the validity of Conjecture \ref{conj:DusartHabib} in the Section \ref{sec:classes}

A graph $G=(V,E)$ is an \emph{interval graph} if there exists a collection of intervals $(I_v)_{v\in V}$ such that $uv\in E$ if and only if the intervals $I_v$ and $I_u$ have non empty intersection. 
Given $G$, such a collection of intervals is not unique and is called an {\em interval representation} of $G$. 
Given an interval representation $\cal R$, one can canonically obtain two orderings of the vertices of $G$: a \emph{left endpoint} ordering of $\cal{R}$ is an ordering of the intervals by increasing value of their left endpoint, and a \emph{right endpoint} ordering of $\cal{R}$ is the ordering of the intervals by decreasing value of  their right endpoint. {\color{black} If some intervals have identical left or right endpoint this can be ambiguous, so more precisely a left (resp. right) endpoint ordering of a collection of intervals $([l(v),r(v)])_{v\in V}$ is any ordering $\prec$ of $V$ such that for all $u,v \in V$, $u\prec v$ implies $l(u)\leq l(v)$ (resp $r(u)\geq r(v)$).}  
It is easy to see that any of these orderings satisfy the following VOC that is in fact a characterization of interval graphs: a graph $G$ is an interval graph if and only if there exists an \emph{I-ordering}, that is an ordering $\sigma$ of $G$ such that :
$$\text{for every triple $a \prec_{\sigma} b \prec_{\sigma} c$, if $ac \in E$ then $ab \in E$}$$
It is a characterization of interval graphs since one can indeed prove that any $I$-ordering is a left endpoint ordering of some interval representation $\cal R$ of $G$. 

An interval graph is a \emph{proper interval graph} if no interval in the interval representation is fully contained in another interval. 
Proper interval graphs were shown in ~\cite{robertsindifference} to be precisely the interval graphs that admit a representation where all the intervals have unit length, and are therefore also called \emph{unit interval graphs}. 
They are also characterized by the following VOC : $G=(V,E)$ is a proper interval graph if and only if $V$ admits a \emph{PI-ordering} :  an ordering $\sigma$ such that 
$$\text{for every triple $a \prec_{\sigma} b \prec_{\sigma} c$, if $ac \in E$ then $ab\in E$ and $bc \in E$.}$$
This VOC follows from the fact that in proper interval graphs, left endpoint and right endpoint orderings are the same. 

A \emph{comparability graph} is a graph $G=(V,E)$ that admits a transitive orientation of its edges. 
That is, there exists an orientation on $E(G)$, where for any triple of vertices $x,y,z$, if $xy, yz \in E(G)$ are oriented $x \rightarrow y$ and $y \rightarrow z$, then the edge $xz$ must exist and is oriented $x \rightarrow z$. 
This transitivity can be captured in a vertex ordering of $V(G)$ known as a \emph{comparability ordering} or a transitive order.
In particular, a transitive order is an ordering $\sigma$ of the vertices of $G$ where if $x \prec_{\sigma} y \prec_{\sigma} z$ and $xy, yz \in E$, then $xz \in E$. 
A \emph{cocomparability graph} is the complement of a comparability graph. 
This definition thus translates into a VOC : a graph $G=(V,E)$ is a cocomparability graph if $V$ admits a so called \emph{cocomparability ordering} (see \cite{KrSt93}), that is an ordering $\sigma$ of $V$ such that
$$\text{for any triple $a \prec_{\sigma} b \prec_{\sigma} c$, if $ac \in E$ then $ab \in E$ or $bc \in E$}$$
For a graph $G$ with an order $\sigma$ on its vertices a triple $a \prec_{\sigma} b \prec_{\sigma} c$ with $ac \in E$, $ab \notin E$ and $bc \notin E$ is called an {\em umbrella}, which is why cocomparability orderings are sometimes called \emph{umbrella-free orderings}.

One can easily see from these vertex orderings that:
\begin{align*}
    \text{Proper Interval} \subsetneq \text{ Interval} \subsetneq \text{Cocomparability}
\end{align*}
It is moreover proven in \cite{gilmore1964characterization} that interval graphs are chordal (no induced cycle of length at least $4$), and even more : they are exactly the $C_4$-free cocomparability graphs. 

Also, it is proved in~\cite{golumbic1984tolerance} that the class of cocomparability graphs are asteroidal triple-free, thus all these graphs have asteroidal number at most two.

Other graph classes we consider in this paper are \emph{domino-free} cocomparability graphs (cocomparability graphs that do not contain the domino as in induced subgraph) and \emph{cobipartite} graphs (the complements of bipartite graphs). 
Since a domino contain a $C_4$ and since interval graphs do not, interval graphs are domino-free. Similarly, a domino contains a independent set of size $3$, so cobipartite graphs form also a subclass of domino-free comparability graph. All inclusions are represented on Figure \ref{fig:classes}.\\

\begin{figure}[htbp]
\begin{center}
    \includegraphics{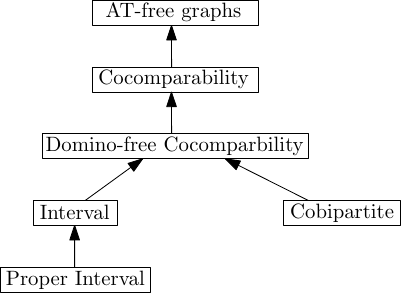}
\caption{Graph classes studied in this article}
\label{fig:classes}
\end{center}
\end{figure}

Vertex orderings produced by searches can also be characterized by vertex orderings (see \cite{CK08} for such results). 
$\LexBFS$ in particular has the following VOC, known as the {\em $\LexBFS$ four point condition}.

 \begin{theorem}\cite{dragan1996lexbfs}(LexBFS 4PC)
 	\label{LexBFScond}
 	Let $G=(V,E)$ be an arbitrary graph. An ordering $\sigma$ is a LexBFS ordering of $G$ if and only if for every triple $a \prec_{\sigma} b \prec_{\sigma} c$, if $ac \in E, ab \notin E$, then there exists a vertex $d$ such that $ d \prec_{\sigma} a$ and $db \in E, dc \notin E$. 
 \end{theorem}
 
 We call the triple $a, b, c$ as described in Theorem \ref{LexBFScond} above a \emph{bad triple}.  
 Observe that the vertex $d$ here is private neighbour of $b$ with respect to $c$. 
 When choosing vertex $d$ as described above, we often choose it as the \emph{left most private neighbour} of $b$ with respect to $c$ in $\sigma$ and write $d = \text{LMPN}(b|_{\sigma} c)$. 
 This is to say that prior to visiting vertex $d$ in $\sigma$, vertices $b$ and $c$ were tied {\color{black} : every vertex before $d$ in $\sigma$ is either a common neighbour or a common non-neighbour of $b$ and $c$} (or equivalently $\lexlabel(b) = \lexlabel(c)$ as assigned by Algorithm~\ref{lexbfsalg}), and vertex $d$ caused $b \prec_{\sigma} c$. 

Combining VOCs for graph classes with the $\LexBFS$ 4PC has already led to a number of structural results~\cite{corneil2013ldfs, kohler2014linear, BigArt}. 
Here we focus on $\LexBFS$ properties on cocomparability graphs. 
In this case, the 4PC can be refined with a stronger statement that we call $C_4$ property.

\begin{property}[The LexBFS $C_4$ Property]
	\label{c4cocomp}
	Let $G=(V, E)$ be a cocomparability graph and $\sigma$ a LexBFS cocomparability order of $V$. If $\sigma$ has a bad LexBFS triple $a \prec_{\sigma} b \prec_{\sigma} c$, then there exists a vertex $d$ such that $d \prec_{\sigma} a$ and $G$ has an induced $C_4 = d, a, b, c$ where $da, db, ac, bc \in E$. 
\end{property}
\begin{proof}
	To see this, it suffices to use the LexBFS 4PC and the cocomparability VOC properties. 
	Since $\sigma$ is a cocomparability ordering, and $ab \notin E$ then $bc \in E$. 
	Then, using the LexBFS 4PC, there must exist a vertex $d \prec a$ such that $db \in E, dc \notin E$. 
	Once again since $d \prec a \prec b$ and $db \in E, ab \notin E$, it follows that $da \in E$ otherwise we contradict $\sigma$ being a cocomparability ordering. 
\end{proof}

We add here another lemma with a flavour similar to the 4PC property, that we will use very often when studying $\LexBFS^+$ multisweep sequences. 
Note that it is true for any graph.

\begin{lemma}
    \label{lem:sameorder}
    Let $G$ be a graph with an ordering $\sigma$ of its vertices and let $\tau=\LexBFS^+(G,\sigma)$. 
    If $a$ and $b$ are vertices such that $a\prec_\sigma b$ and $a\prec_\tau b$, then there exists a vertex $c$ with $c\prec_\tau a$ such that $ca\in E$ and $cb\notin E$. 
    Furthermore, if $c$ is the leftmost vertex for this property (i.e. $c=\text{LMPN}(a|_{\tau} b)$), then every vertex that precedes $c$ in $\tau$ is either adjacent to both $a$ and $b$ or to none of them.
\end{lemma}
\begin{proof}
    This is just the consequence of the $^+$ rule: if $a$ precedes $b$ in both orderings, then it means $a$ and $b$ were not tied when $a$ was picked during the construction of $\tau$, and therefore the label of $a$ was strictly larger than the one of $b$, which exactly translates into the conclusion of the Lemma.
\end{proof}

A consequence of the previous lemma is a result from \cite{BigArt} known as the \emph{Flipping Lemma}, that gives an intuition as to why Conjecture~\ref{conj:DusartHabib} could be true.
\begin{lemma}[The Flipping Lemma,\cite{BigArt}]\label{flippinglemma}
	Let $G=(V,E)$ be a cocomparability graph, $\sigma$ a cocomparability ordering of $G$ and $\tau = \text{LexBFS}^+(\sigma)$. 
	For every pair $u,v$ such that $uv \notin E$, $u \prec_{\sigma} v$ if and only if $v \prec_{\tau} u$.
\end{lemma}
\begin{proof}
    Assume by contradiction that there exists vertices $u$ and $v$ such that $u \prec_{\sigma} v$ and $u \prec_{\tau} v$, and choose such a pair with the left most possible element $u$ with respect to $\tau$. By Lemma \ref{lem:sameorder}, there exists a vertex $w$ such that $w\prec_{\tau} u$, $wu\in E$ and $wv \notin E$. 
    Because of the choice of the pair $(u,v)$, we must have $v\prec_{\sigma} w$, but now the triple $(u,v,w)$ forms an umbrella in $\sigma$, which contradicts the fact that $\sigma$ is a cocomparability order on $G$.
\end{proof}

Given that a comparability ordering is an umbrella-free ordering, the Flipping Lemma directly implies the following result of \cite{BigArt}, which states that LexBFS$^+$ sweeps preserve cocomparability orderings.

\begin{theorem}\cite{BigArt}
\label{maintaincocomp}
	Let $\sigma$ be a cocomparability ordering of $G=(V,E)$. The 
	ordering $\tau = \LexBFS^+(\sigma)$ is a cocomparability ordering of $G$. 
\end{theorem}

Another easy consequence of the Flipping Lemma is the following corollary. 

\begin{corollary}
\label{evenlength}
	For a non-trivial cocomparability graph $G$ (i.e. $|V(G)| \ge 2$), $\LexCycle(G)$ is necessarily even.
\end{corollary}
\begin{proof}
	If $G$ contains a pair of nonadjacent vertices, then the claim is a trivial consequence of the 
	Flipping Lemma. Otherwise $G$ is a complete graph and $\sigma_2=\sigma_1^d$
	is the cycle of length 2.
\end{proof}
 An example of a graph which illustrates that this is not the case for all graphs is the graph $G_3$ with $\LexCycle(G_3)=3$ drawn in Figure~\ref{G3}.\\

 If Conjecture~\ref{conj:DusartHabib} is true, then Theorems~\ref{thm:DusartHabib} and~\ref{maintaincocomp} together imply that for any starting ordering $\sigma_0$, a $\LexBFS^+$ multisweep on a cocomparability graph $G$ always ends on a $2$-cycle consisting of two cocomparability orderings of $G$. 
 Therefore, if Conjecture~\ref{conj:DusartHabib} is true, we would have the following simple algorithm for getting a transitive orientation  of a comparability graph.

\begin{algorithm}[H]
	\caption{A \emph{Potential} Simple Transitive Orientation Algorithm}
	\label{transitive}
	\begin{algorithmic}[1]
		\Input A comparability graph $G=(V,E)$
		\Output A comparability order of $G$
		\State Construct $G'$ the complement of $G$
		\State Take an arbitrary order $\sigma_0$ on the vertices of $G'$
		\State $\sigma_1 \leftarrow \LexBFS^+(G',\sigma_0)$,  $\sigma_2 \leftarrow \LexBFS^+(G',\sigma_1)$
		\State $i \leftarrow 2$
		\While{$\sigma_i \neq \sigma_{i-2}$}
		\State $i \leftarrow i+1$
		\State $\sigma_i \leftarrow \LexBFS^+(G',\sigma_{i-1}$)
		\EndWhile
		\State \Return $\sigma_i$
	\end{algorithmic}
\end{algorithm}

\section{Domino-free Cocomparability Graphs}
\label{sec:classes}

In support of Conjecture~\ref{conj:DusartHabib}, we show in this section that the conjecture holds for the subclass of domino-free cocomparability graphs. 
This class in particular includes the classes of proper interval, interval and cobipartite graphs, but for these three subclasses we provide independent proofs which imply stronger results. 
For interval graphs we show that the two orderings of the $\LexCycle$ are left endpoint and right endpoint orderings of the \emph{same} interval representation, and that such a cycle is reached in at most $n$ iterations of the multisweep algorithm. 
Moreover in the case of proper interval graphs, we prove that the cycle is reached in at most 3 iterations and that the 2 cycles are duals one of another. 
The independent proof for cobipartite graphs is, first of all, interesting for the different flavor of the proof, and secondly it provides an upper bound of $3n$ iterations of multisweep algorithm before reaching the cycle.

\label{ssec:domino}
\subsection{Domino-free cocomparability graphs}
Here we prove the more general result of the paper regarding Conjecture \ref{conj:DusartHabib}. 
Recall that a domino is the graph obtained from a cycle of length $6$ by adding a diametral chord (see Figure \ref{fig:domino}).

\begin{theorem}\label{thm:domino}
	Domino-free cocomparability graphs have $\LexCycle = 2$.
\end{theorem}

\begin{proof}
	Let $G=(V, E)$ be a domino-free cocomparability graph. Let $\sigma_1, \ldots, \sigma_k$ be a LexBFS$^+$ cycle obtained by a multisweep $\LexBFS^+$ process on $G$, and assume by contradiction that $k > 2$. 
	Recall that by Corollary \ref{evenlength}, $k$ is even and also because of Theorem \ref{thm:DusartHabib} and Theorem \ref{maintaincocomp}, we can assume that every $\sigma_i$ is a cocomparability ordering. 	For two consecutive orderings of the same parity (index $i$ is considered mod $k$) :
	\begin{align*}
	    \sigma_i = u_1, u_2, \ldots, u_n \text{ \hspace{0.2cm} and \hspace{0.2cm}} \sigma_{i+2} = v_1, v_2, \ldots, v_n 
	\end{align*}
	let $\text{diff}(i)$ denote the index of the first (left most) vertex that is different in $\sigma_i, \sigma_{i+2}$:
	\begin{align*}
	    \diff(i) =\min \{ j \in [n]  \mid u_j \neq v_j \}
	\end{align*}
	
	Now up to ``shifting" the start of the cycle, we can assume without loss of generality that $\diff(1)$ is minimal amongst all $\diff(i)$. 
	Also from now on, in order to use lighter notations, we will write $\prec_i$ instead of $\prec_{\sigma_i}$, and $LMPN(x|_{k}y)$ instead of $LMPN(x|_{\sigma_k}y)$.

	Let then $a, b$ be the first (left most) difference between $\sigma_1$ and $\sigma_3$. 
	Denoting $\sigma_1 = u_1, u_2, \ldots, u_n$ and $\sigma_3 = v_1, v_2, \ldots, v_n$, and $j = \text{diff}(1)$, we have thus $u_i = v_i, \forall i < j$ and $u_j = a, v_j = b$. 
	Note that this implies in particular $a \prec_1 b$ and $b \prec_3 a$. 
	Furthermore, if we define $S = \{u_1, \ldots, u_{j-1}\} = \{v_1, \ldots, v_{j-1}\}$, then $\sigma_1[S] = \sigma_3[S]$, so at the time $a$ (resp. $b$) was chosen in $\sigma_1$ (resp. $\sigma_3$), $b$ (resp. $a$) had the same label. 
	Therefore in both cases it means the $^+$ rule was applied to break ties between $a$ and $b$ and so $b \prec_{k} a$ and $a \prec_2 b$. 
	We thus have :
	
    \begin{figure}[H]
    \begin{center}
        \includegraphics[page=1, scale = 1]{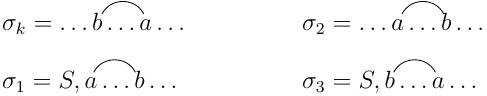}
    \end{center}
    \end{figure}

	Since $a \prec_{1} b$ and $a \prec_{2} b$, Lemma~\ref{lem:sameorder} applies, so we choose vertex $c$ as $c = \text{LMPN}(a |_2 b)$.
	Using the Flipping Lemma on $b$ and $c$, we place vertex $c$ in the remaining orderings as follows:

    \begin{figure}[H]
    \begin{center}
        \includegraphics[page=2, scale = 1]{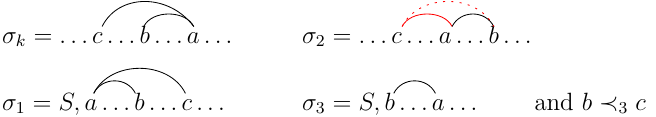}
    \end{center}
    \end{figure}

	This gives rise to a bad LexBFS triple in $\sigma_k$ where $c \prec_k b \prec_k a$ and $ca \in E, cb \notin E$. By the LexBFS $C_4$ Property \ref{c4cocomp}, there exists a vertex $d \prec_{k} c$ such that $d = \text{LMPN}(b |_{k} a)$ and $dc\in E$.
	 We again use the Flipping Lemma for $ad \notin E$ to place $d$ in the remaining orderings. Note that in $\sigma_2$, the Flipping Lemma places $d \prec_2 a$, and by the choice of $c$ as $\text{LMPN}(a|_{2}b)$, it follows that no private neighbour of $b$ with respect to $a$ could be placed before $c$ in $\sigma_2$. 
	 Therefore we can conclude that $c \prec_2 d \prec_2 a$. 

    \begin{figure}[H]
    \begin{center}
        \includegraphics[page=3, scale = 1]{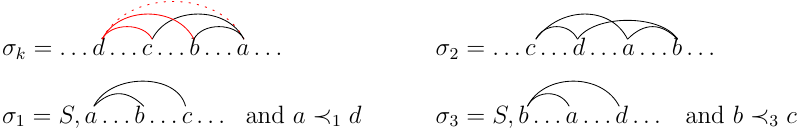}
    \end{center}
    \end{figure}

	It remains to place $d$ in $\sigma_1$ and $c$ in $\sigma_3$. 
	We start with vertex $d$ in $\sigma_1$. 
	We know that $a \prec_1 d$. 
	This gives rise to three cases: Either \textbf{(i)} $c \prec_1 d$, or \textbf{(ii)} $a \prec_1 d \prec_1 b$, or \textbf{(iii)} $b \prec_1 d \prec_1 c$. 
	
	\textbf{(i).} If $c\prec_1 d$ then since $c \prec_{2} d$, so we apply Lemma \ref{lem:sameorder} and choose a vertex $e$ as $e = \text{LMPN}(c|_{2}d)$. 
	This means $ed \notin E$, and since $da \notin E$ and $e \prec_2 d \prec_2 a$, it follows that $ea \notin E$ for otherwise the triple $e,d,a$ would form an umbrella.
    
    \begin{figure}[H]	
    \begin{center}
        \includegraphics[page=4, scale = 1]{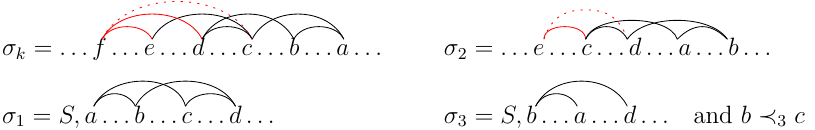}
    \end{center}
    \end{figure}
	
	Furthermore, by the choice of vertex $c$ as $\text{LMPN}(a|_{2}b)$, and the facts that $e \prec_2 c$ and $ea \notin E$, it follows that $eb \notin E$, otherwise $e$ would be a private neighbour of $b$ with respect to $a$ that is to the left of $c$ in $\sigma_2$. 
	Using the Flipping Lemma, we place vertex $e$ in the remaining orderings, and in particular, placing vertex $e$ in $\sigma_k$ gives rise to a bad LexBFS triple $e, d, c$. 
	By the LexBFS 4PC and the LexBFS $C_4$ Property, there must exist a vertex $f$ chosen as $f = \text{LMPN}(d|_{k} c)$ and $fe \in E$. 
	Using the same argument above, one can show that $fc \notin E$ and $cb \notin E$ implies $fb \notin E$, and given the choice of $d$ in $\sigma_1$ and $fb \notin E$, then $fa \notin E$. 
	We, therefore, have the induced domino $abcdef$. A contradiction to $G$ being domino-free.

	\textbf{(ii).} If $a \prec_1 d \prec_1 b$, then $a, d, b$ forms a bad LexBFS triple, and thus by Theorem \ref{LexBFScond}, choose vertex $e\prec_1 a$ as $e = \text{LMPN}(d|_1b)$, therefore $eb \notin E$. 
	By the $C_4$ property (Property \ref{c4cocomp}), $ea \in E$. 
	Since $e \prec_1 a$, it follows $e \in S$. 
	But then $ea \in E, eb \notin E$ implies $\lexlabel(a) \neq \lexlabel(b)$ when $a,b$ were chosen. 
	A contradiction to $S \cap N(a) = S \cap N(b)$. 
	
	\begin{figure}[H]	
    \begin{center}
        \includegraphics[page=5, scale = 1]{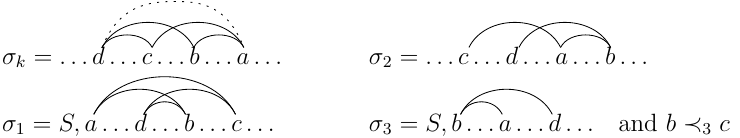}
    \end{center}
    \end{figure}

	\textbf{(iii).} We thus must have $b \prec_1 d \prec_1 c$, in which case we still have a bad LexBFS triple given by $a, d, c$ in $\sigma_1$. 
	Choose vertex $e \prec_1 a$ as $e = \text{LMPN}(d|_1c)$ {\color{black} (and remember for later that as explained after Theorem \ref{LexBFScond}, $e$ is such that every vertex placed before $e$ is either a common neighbour or a common non-neighbour of $c$ and $d$)}. 
	By property \ref{c4cocomp}, $ea \in E$, and since $e \prec_1 a$, it follows $e \in S$, and thus $eb \in E$ since $S\cap N(a) = S \cap N(b)$. 
	Since $\sigma_1[S] = \sigma_3[S]$, it follows that $e$ appears in $\sigma_3$ in $S$, and thus $e$ is the LMPN$(d|_3c)$ as well. Therefore $d \prec_3 c$. 
	The orderings look as follows: 

    \begin{figure}[H]	
    \begin{center}
        \includegraphics[page=6, scale = 1]{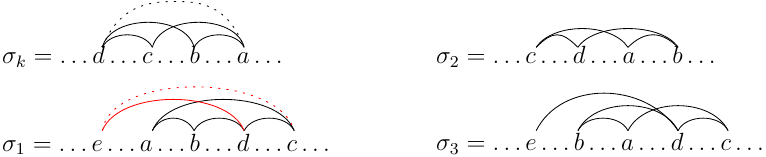}
    \end{center}
    \end{figure}

	Consider the ordering of the edge $cd$ in $\sigma_{k-1}$. 
	If $d \prec_{k-1}c$, we use the same argument above to exhibit a domino as follows: if $d \prec_{k-1}c$, then $d \prec_{k-1,k} c$, so choose a vertex $p = \text{LMPN}(d|_kc)$. 
	Therefore $pc \notin E$, and since $cb \notin E$ and $p \prec_k c\prec_k b$, it follows that $pb \notin E$ as otherwise we contradict $\sigma_k$ being a cocomparability ordering. 
	Moreover, given the choice of vertex $d$ in $\sigma_k$ as the LMPN$(b|_ka)$ and the fact that $p \prec_k d, pb \notin E$, it follows that $pa \notin E$ as well. 
	We then use the Flipping Lemma to place vertex $p$ in $\sigma_2$. 
	This gives rise to a bad LexBFS triple $p, c, d$ in $\sigma_2$. 
	Choose vertex $q \prec_2 p$ as $q = \text{LMPN}(c|_2d)$. 
	Again, one can show that $qa, qb \notin E$, and thus the $C_4$s in $\{a, b, c, d, p, q\}$ are induced, therefore giving a domino; a contradiction to $G$ being domino-free.
	
	Therefore we must have $c \prec_{k-1} d$. 
	Consider now the first (left most) difference between $\sigma_{k-1}$ and $\sigma_1$. 
	Let $S'$ be the set of initial vertices that is the same in $\sigma_{k-1}$ and $\sigma_1$. 
	By the choice of $\sigma_1$ as the start of the cycle $\sigma_1, \sigma_2, \ldots, \sigma_k$, and in particular as the ordering with minimum diff$(1)$, we know that $|S| \le |S'|$. 
	Since $S$ and $S'$ are both initial segments  of $\sigma_1$, it follows that $S \subseteq S'$, and the ordering of the vertices in $S$ is the same in $S'$ in $\sigma_1$; $\sigma_1[S] \subseteq \sigma_1[S']$. 
    {\color{black}In particular vertex $e$ as constructed above appears in $S'$ as the left most private neighbour of $d$ with respect to $c$ in $\sigma_1$, and every vertex before $e$ in $\sigma_{k-1}$ is either a common neighbour of $c$ and $d$ or a common non-neighbour of $c$ and $d$. But then $d$ must have been chosen before $c$, which contradicts $c \prec_{k-1} d$.}

	Notice that in all cases, we never assumed that $S \neq \emptyset$. 
	The existence of an element in $S$ was always forced by bad LexBFS triples. 
	If $S$ was empty, then case \textbf{(i)} would still produce a domino, and cases \textbf{(ii), (iii)} would not be possible since $e \in S$ was forced by LexBFS. 
	
	To conclude, if $G$ is a domino-free cocomparability graph, then  it cannot have $\LexCycle(G) > 2$.  
\end{proof}

\subsection{Interval graphs}
For the special case of interval graphs, we prove a stronger statement about the $2$-cycle: it is reached almost as soon as one gets a cocomparability order, and furthermore the two orderings are left and right endpoint ordering of the same interval representation.

\begin{theorem}
	\label{thm:interval}
	Let $G$ be an interval graph with $|V(G)| > 1$, $\sigma_0$ an arbitrary {\color{black}LexBFS} cocomparability order of $G$ and $\{ \sigma_i, \}_{i\geq 1} $ a sequence of LexBFS$^+$ orderings where $\sigma_i = \text{LexBFS}^+(\sigma_{i-1})$. 
	Then the following properties hold :
	\begin{itemize}
	    \item $\sigma_1 = \sigma_3$.
	    \item There exists an interval representation $\cal R$ of $G$ such that $\sigma_1$ and $\sigma_2$ are respectively a left endpoint ordering and a right endpoint ordering of $\cal R$.
	\end{itemize}
\end{theorem}

Before giving the proof of the theorem, let us observe that by Theorem \ref{thm:DusartHabib}, in any multisweep $\LexBFS^+$ sequence such an order $\sigma_0$ is reached in at most $n$ steps, if $n$ is the number of vertices of the graph. 
Consequently we have that the $2$-cycle is reached in at most $n+1$ steps for interval graphs.

Moreover, the second item above implies in particular that $\sigma_1$ and $\sigma_2$ are $I$-orderings.
This is in fact guaranteed by the following easy lemma.

\begin{lemma}\label{lem:LexBFS+Cocomp=LeftPoint}
    Let $G$ be an interval graph, and $\sigma$ a cocomparability ordering of $G$. Then $\tau = \LexBFS^+(G, \sigma)$ is an I-ordering of $G$. 
\end{lemma}
\begin{proof}
    Assume by contradiction $\tau$ is not an I-ordering. 
    Then there exists a a triple $a \prec_{\tau} b \prec_{\tau} c$ where $ac \in E$ and $ab \notin E$. 
    Thus the triple $abc$ forms a bad triple in $\tau$ and thus by the $\LexBFS C_4$ property (Property~\ref{c4cocomp}), there exists a vertex $d \prec_{\tau} a$ such that $d,a,b,c$ induces a $C_4$ in $G$, a contradiction to $G$ being chordal, and thus interval. 
\end{proof}

Here is a second lemma that will imply the second item of the Theorem.
\begin{lemma}
\label{lem:intervalrep}
    Let $G$ be an interval graph, and $\sigma$ an $I$-ordering of $G$. If $\tau = \LexBFS^+(G, \sigma)$, then there exists an interval representation $\cal R$ of $G$ such that $\sigma$ and $\tau$ are respectively the left endpoint ordering and the right endpoint ordering of $\cal R$.
\end{lemma}
\begin{proof} 
    Recall that formally $\sigma$ and $\tau$ are bijections from $V$ to $\{1,\ldots,n\}$. 
    Define $f_\sigma:V\to \{1,\ldots,n\}$ by
      \[
       f_\sigma(v) = \max\{\sigma(w) \mid v\prec_{\sigma} w \text{ or } w=v\}
      \]
    Informally, $f_\sigma(v)$ is the position in $\sigma$ of the rightmost neighbour of $v$ to the right of $v$, or $\sigma(v)$ if there is no such neighbour. 
    
    For every vertex $v$, define the interval $I_v=[\sigma(v),f_\sigma(v)]$ and call $\cal R$ the resulting collection. Let us prove first that  that $\cal R$ is indeed an interval representation of $G$. Let $u, v$ be two vertices and assume without loss of generality that $u\prec_{\sigma}v$ (that is $\sigma(u)< \sigma(v)$). If $uv\in E$, then by definition $\sigma(v)\leq f_\sigma(u)$ so that $I_u$ and $I_v$ both contain $\sigma(v)$. Conversely if $uv\notin E$, then because $\sigma$ is an $I$-ordering, there is no neighbour of $u$ that is placed after $v$ in $\sigma$, so we have $f_\sigma(u)<\sigma(v)$, and therefore $I_u$ and $I_v$ are disjoint as required.

By definition $\sigma$ is a left point ordering of $\cal R$, so to conclude we have to prove that $\tau$ is a right endpoint ordering of $\cal R$, that is for any vertices $u$ and $v$, $f_\sigma(u)>f_\sigma(v)$ implies $\tau(u)<\tau(v)$. The inequality on $f_\sigma$ implies that : either $f_\sigma(u)=\sigma(u)$ and therefore $v\prec_{\sigma} u$ and $vu\notin E$, or $f_\sigma(u)>\sigma(u)$ and thus there exists $w$ placed after $u$ and $v$ in $\sigma$ such that $vw\notin E$ and $uw\in E$. But since $\sigma$ is a cocomparability ordering, the Flipping Lemma \ref{flippinglemma} applies : in the first case we directly get $u\prec_{\tau} v$ and in the second one we first have $w\prec_{\tau}v$, which, since $\tau$ is an I-ordering, also implies that $u\prec_{\tau} v$, as required.

\end{proof}
We are now ready for the proof of the main theorem of this subsection.
\begin{proof}[Proof of Theorem \ref{thm:interval}]
    Note that the second item follows directly from Lemma \ref{lem:LexBFS+Cocomp=LeftPoint} (applied to $\sigma=\sigma_0$) and Lemma \ref{lem:intervalrep} (applied to $\sigma=\sigma_1$). 
    Let us thus now prove the first item.
    Consider the following orderings:
    \begin{align*}
        \sigma_1 = \text{ LexBFS}^+(\sigma_0) \qquad 	\sigma_2 = \text{ LexBFS}^{+}(\sigma_1) \qquad 	\sigma_3 = \text{ LexBFS}^{+}(\sigma_2)
    \end{align*}
    
    Suppose, for sake of contradiction, that  $\sigma_1 \neq \sigma_3$. 
    Let $k$ denote the index of the first (left most) vertex where $\sigma_1$ and $\sigma_3$ differ. 
    In particular, let $a$ (resp. $b$) denote the $k^{\text{th}}$ vertex of $\sigma_{1}$ (resp. $\sigma_{3}$). 
    Let $S$ denote the set of vertices preceding $a$ in $\sigma_{1}$ and $b$ in $\sigma_{3}$. 
    
    Since the ordering of the vertices of $S$ is the same in both $\sigma_{1}$ and $\sigma_{3}$, and $a,b$ were chosen in different LexBFS orderings, it follows that $\lexlabel(a) = \lexlabel(b)$ in both $\sigma_1$ and $\sigma_3$ when both $a$ and $b$ were being chosen. 
    Therefore, $N(a) \cap S = N(b) \cap S$. So if $a$ were chosen before $b$ in $\sigma_1$ then the $^+$ rule must have been used to break ties between $\lexlabel(a) = \lexlabel(b)$. 
    This implies $b \prec_0 a$, similarly $a \prec_2 b$. 
    The ordering of the pair $a,b$ is thus as follows: 
    \begin{align*}
        \sigma_0 : \hspace{0.5cm}& \ldots b \ldots a \ldots &\sigma_2 : \hspace{0.5cm}& \ldots a \ldots b \ldots\\
        \sigma_1 : \hspace{0.5cm}& \ldots a \ldots b \ldots  &\sigma_3 : \hspace{0.5cm}& \ldots b \ldots a \ldots 
    \end{align*}
    Using the Flipping Lemma, it is easy to see that $ab \in E$. 
    Since $a\prec_{1,2} b$, we can apply Lemma \ref{lem:sameorder} and choose a vertex $c$ as $c = \text{LMPN}(a|_2b)$. 
    Therefore $c \prec_2 a \prec_2 b$ and $ac \in E, bc \notin E$. 
    
    Since $\sigma_0$ is a cocomparability order, by Theorem \ref{maintaincocomp}, $\sigma_1, \sigma_2, \sigma_3$ are cocomparability orderings. 
    Using the Flipping Lemma on the non-edge $bc$, we have $c \prec_2 b$ implies $c \prec_0 b$. 
    Therefore in $\sigma_0$, $c \prec_0 b \prec_0 a$ and $ac \in E, bc \notin E$. 
    Using the LexBFS 4PC (Theorem \ref{LexBFScond}), there exists a vertex $d$ in $\sigma_0$ such that $d \prec_0 c \prec_0 b \prec_0 a$ and $db \in E, da \notin E$. 
    By the LexBFS $C_4$ cocomparability property (Property \ref{c4cocomp}), $dc \in E$ and the quadruple $abdc$ forms an induced $C_4$ in $G$, thereby contradicting $G$ being an interval graph. 
\end{proof}
\subsection{Proper Interval Graphs} 
For proper interval graphs, Corneil proved the following result, which is stronger than Theorem~\ref{thm:DusartHabib}:
\begin{theorem}\cite{Corneil04}\label{3sweep}
	A graph $G$ is a proper interval graph if and only if the third LexBFS$^+$ sweep on $G$ is a PI-ordering. 
\end{theorem}

We already know by Theorem \ref{thm:interval} that the $2$ cycle is reached one step after reaching an $I$-order. For proper interval we prove additionally that the $2$ orderings in a $2$-cycle are duals one of another.
\begin{theorem}
	\label{thm:properinterval}
	Let $G$ be a proper interval graph and $\sigma$ a PI-ordering of $G$, then LexBFS$^+(\sigma)=\sigma^d$.
\end{theorem}
\begin{proof}
    Define $\tau=\text{LexBFS}^+(\sigma)$. 
    All we have to prove is that for any vertices $x\prec_{\sigma}y$ implies $y\prec_{\tau}x$. 
    For non edges this is exactly Flipping Lemma \ref{flippinglemma}, so we can assume that $xy\in E$. 
    Assume by contradiction that $x\prec_{\tau}y$. 
    Since the pair maintained the same order on consecutive sweeps, we can apply Lemma \ref{lem:sameorder} to get a vertex $z$ such that $z \prec_{\tau} x \prec_{\tau} y$ and $zx \in E, zy \notin E$.
    Using the Flipping Lemma, this implies $x\prec_{\sigma}y\prec_{\sigma}z$  with $xy,xz \in E$ and $yz \notin E$, which contradicts $\sigma$ being a PI-ordering. 
\end{proof}

Therefore, using Theorem \ref{thm:properinterval} and Theorem \ref{3sweep}, we get Corollary~\ref{concludepi}.

\begin{corollary}
\label{concludepi}
	If $G$ is a proper interval graph with $|V(G)| > 1$, Algorithm \ref{transitive} stops at $\sigma_5=\sigma_3$, if not sooner.
\end{corollary}
\begin{proof}
	By Theorem \ref{3sweep}, we know that $\sigma_3$ is a PI-ordering. 
	Using Theorem \ref{thm:properinterval}, we conclude that Algorithm \ref{transitive} applied on a PI-ordering computes $\sigma_4=\sigma_3^d$ and $\sigma_5=\sigma_4^d= (\sigma_3^d)^d=\sigma_3$.
\end{proof}
\subsection{Cobipartite Graphs}
\label{cobip}
In this section we study cobipartite graphs, i.e. graphs whose vertex set can be partitioned into two cliques. 
These are clearly domino-free (as the complement of a domino contains a triangle), so the fact that such graphs have LexCycle equal to $2$ is a consequence of Theorem \ref{thm:domino}. 
In this section we give a separate proof of this result that we think is interesting for three reasons : 
\begin{itemize}
    \item We prove that the cycle is reached in at most $3n$ sweeps.
    \item We prove that the cycle is in fact composed of an order and its dual.
    \item The proof technique sheds light on the link between this problem and the one of doubly lexicographic orderings on rows and columns of matrices.
\end{itemize}

Let $G=(V = A \cup B, E)$ be a cobipartite graph, where both $A$ and $B$ are cliques. 
Notice that any ordering $\sigma$ on $V$ obtained by first placing all the vertices of $A$ in any order followed by the vertices of $B$ in any order is a cocomparability ordering. 
In particular, such an ordering is precisely how any LexBFS cocomparability ordering of $G$ is constructed, as shown by Lemma \ref{2cliques} below. We first show the following easy observation. 

\begin{lemma}
\label{adjcob}
	Let $G$ be a cobipartite graph, and let $\sigma$ be a cocomparability ordering of $G$. 
	In any triple of the form $a \prec_{\sigma} b \prec_{\sigma} c$, either $ab \in E$ or $bc \in E$. 
\end{lemma}
\begin{proof}
	Suppose otherwise, then if $ac \in E$, we contradict $\sigma$ being a cocomparability ordering, and if $ac \notin E$, then the triple $abc$ forms a stable set of size 3, which is impossible since $G$ is cobipartite. 
\end{proof}
\begin{lemma}\label{2cliques}
	Let $G$ be a cobipartite graph, and let $\sigma = x_1,x_2,\ldots x _n$ be a LexBFS cocomparability ordering of $G$. There exists $i \in [n]$ such that $\{x_1,\ldots,x_i\}$ and $\{x_{i+1},\ldots,x_n\}$ are both cliques.
\end{lemma}
\begin{proof}
	Let $i$ be the largest index in $\sigma$ such that $\{x_1,\ldots,x_i\}$ is a clique. 
	Suppose $\{x_{i+1}, \ldots, x_n\}$ is not a clique, and consider a pair of vertices $x_j, x_k$ where $x_jx_k \notin E$ and $i+1 \le j < k$. 
	By the choice of $i$, vertex $x_{i+1}$ is not universal to $\{x_1,\ldots,x_i\}$. 
	Since $\sigma$ is a LexBFS ordering, vertex $x_j$ is also not universal to $\{x_1,\ldots,x_i\}$ for otherwise $\lexlabel(x_j)$ would be lexicographically greater than $\lexlabel(x_{i+1})$ implying $j < i+1$ -  unless $i+1 = j$, in which case $x_{j}$ is $x_{i+1}$ and we just showed that $x_{i+1}$ is not universal to $\{x_1,\ldots,x_i\}$, thus $x_j$ is also not universal to $\{x_1,\ldots,x_i\}$. 
	Let $x_p \in \{x_1, \ldots, x_i\}$ be a vertex not adjacent to $x_j$. 
	We thus have $x_p \prec_{\sigma} x_j \prec_{\sigma} x_k$ and both $x_px_j, x_jx_k \notin E$. 
	A contradiction to Lemma \ref{adjcob} above.
\end{proof}

Since cobipartite graphs are cocomparability graphs, by Theorem \ref{thm:DusartHabib}, after a certain number $t \le n$ of iterations, a series of $\LexBFS^+$ sweeps yields a cocomparability ordering $\sigma_t$. 
By Lemma \ref{2cliques}, this ordering consists of the vertices of one clique $A$ followed by another clique $B$. 

Assume $a_1,\ldots,a_p,b_q,\ldots,b_1$ is the ordering of $\sigma_t$ (the reason why the indices of $B$ are reversed will be clear soon). Consider the $p\times q$ matrix $M$ defined as follows:

    \[
    M_{i,j}=
        \begin{cases}
            &\mbox{$1$ if $a_i b_j \in E$}\\
            &\mbox{$0$ otherwise}
        \end{cases}
    \]
(All through this section, and for any matrix $A$, we will denote by $A_{i,j}$ the coefficient of $A$ on row $i$ and column $j$.)

The easy but crucial property that follows from the definition of LexBFS is the following: the columns of this matrix $M$ are sorted lexicographically in increasing order (for any pair of vectors of the same length $X$ and $Y$, lexicographic order is defined by $X<_{lex} Y$ if the least integer $k$ for which $X_k\neq Y_k$ satisfies $X_k<Y_k$). 

Consider $\sigma_{t+1} = \text{LexBFS}^+(\sigma_t)$, and notice that $\sigma_{t+1}$ begins with the vertices of $B$ in the ordering $b_1,b_2,\ldots,b_q$ followed by the vertices of $A$ which are sorted exactly by sorting the corresponding rows of $M$ lexicographically in non-decreasing order (the first vertex to appear after $b_q$ being the maximal row, that is the one we put at the bottom of the matrix). 
But then to obtain $\sigma_{t+2}$ we just need to sort the columns lexicographically, and so on.

Therefore to prove that $\LexCycle = 2$ for cobipartite graphs, it suffices to show that this process must converge to a fixed point: that is, after some number of steps, we get a matrix such that both rows and columns are sorted lexicographically, which implies we have reached a $2$ cycle. 
This is guaranteed by the following Proposition (which we state for $0-1$ matrices, but the proof works identically for any integer valued matrix).

\begin{proposition}
\label{prop:matrix}
    Let $M$ be a $p\times q$ matrix with $\{0,1\}$ entries. 
    Define two sequences of matrices $(R^{(t)})_{t\ge 0}$ and $(C^{(t)})_{t\ge 1}$ as follows:
    \begin{itemize}
        \item $R^{(0)}=M$
        \item For $t\geq 1$, $C^{(t)}$ is obtained by sorting the columns of $R^{(t-1)}$ in non-decreasing lexicographical order.
        \item For $t\geq 1$, $R^{(t)}$ is obtained by sorting the rows of $C^{(t)}$ in non-decreasing lexicographical order.
    \end{itemize}
    Then there exists $k \le q$ for which $R^{(k)}=C^{(k)}$
\end{proposition}  
Note that the conclusion in fact implies that both sequences are constant from the index $k$ since this implies that $R^{(k)}$ has both its rows and columns sorted lexicographically. 
This proposition is reminiscent of the doubly lexical ordering of $\{0,1\}$ matrices studied by Lubiw in~\cite{lubiw1987doubly}. 
What we prove implies that in order to obtain this doubly lexical ordering, one can do in fact a sequence of at most $n$ LexBFS, giving thus a $O(nm)$ time, if $m$ denotes the number of non zero entries. 
Better algorithms for this problem are already known. 
For instance, in~\cite{spinradDoublyn2}, Spinrad gave an $O(n^2)$ time for dense matrices. 
\begin{proof}[Proof of Proposition \ref{prop:matrix}] 
    We rely on the following claim.\\
    
    {\bf \noindent Claim :} For every $t$, the $t$ first columns of $R^{(t)}$ are sorted in non decreasing lexicographic order, and are smaller than the $q-t$ last columns of the matrix.\\
         
    This indeed implies the desired result, as for $t=q$ we have that the whole matrix $R^{(q)}$ is doubly lexicographic (rows by construction and columns follow from the claim), which proves our Proposition. 
       
    We prove the Claim by induction on $t$.     
    This is obvious for $t=0$, so let us assume that $t\geq 1$ and that the property is true for $t-1$. 
    The induction hypothesis implies that in $C^{(t)}$, the $(t-1)$-first columns are identical to those of $R^{(t-1)}$. 
    Since the rows of $R^{(t-1)}$ were sorted (by definition), we have that the matrix resulting from $C^{(t)}$ by looking just at the first $(t-1)$-columns is sorted both for rows and columns. 
    Therefore in $R^{(t)}$, these same entries will also be identical and so the $(t-1)$ first columns of $R^{(t)}$ are sorted.  
 
    Assume by contradiction that for some $p\leq t$ and $q\geq t$, the $q$-th column of $R^{(t)}$ is strictly smaller than the $p$-th column. 
    Let then $i$ be the smallest integer such that $R^{(t)}_{i,p}\neq R^{(t)}_{i,q}$, and thus $R^{(t)}_{i,p}=1$ and $R^{(t)}_{i,q}=0$. 
    Since $R^{(t)}$ was obtained from $C^{(t)}$ by a reordering of its rows, there exists $i'$ such that  $C^{(t)}_{i',p}=1$ and $C^{(t)}_{i',q}=0$. 
    
    {\color{black}
    Since the columns in $C^{(t)}$ are sorted, we deduce that column $p$ in $C^{(t)}$ is lexicographically smaller than column $q$ in $C^{(t)}$.
    Since $C^{(t)}_{i',p}=1$ and $C^{(t)}_{i',q} = 0$, there must exist $j' < i'$ such that $C^{(t)}_{j',p} = 0$ and $C^{(t)}_{j',q} = 1$ as illustrated below. 
    
    Row $j'$ must appear somewhere in $R^{(t)}$, which was obtained by sorting rows of $C^{(t)}$.
    Recall, however, that first $(t-1)$ columns of $C^{(t)}$ have their rows sorted.
    And thus in particular, the first $(t-1)$ elements of the $j'^{th}$ row of $C^{(t)}$ are the same or lexicographically smaller than the first $(t-1)$ elements of the $i'^{th}$ row of $C^{(t)}$. 
    This means that sorting the rows of $C^{(t)}$ puts row $j'$ above row $i'$.
    Since row $i'$ lands at row $i$ in $R^{(t)}$, row $j'$ must land above row $i$ in $R^{(t)}$. 
    But all rows above row $i$ in $R^{(t)}$ have identical element in the $p^{th}$ and $q^{th}$ columns by the minimality of $i$, while $j'$ does not, a contradiction. 
    }
 
    This concludes the induction.
    
    \begin{figure}[htbp]
    \begin{center}
        \includegraphics[scale=0.75]{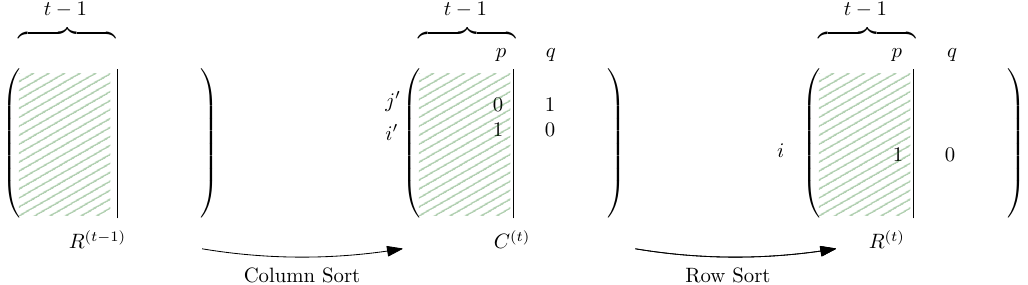}
    \caption{The matrices in Proposition \ref{prop:matrix}}
    \label{fig:matrix}
    \end{center}
    \end{figure}
\end{proof}
We conclude with the following corollary:
\begin{corollary}
	    Cobipartite graphs have $\LexCycle = 2$, this cycle is reached in fewer than $3n \LexBFS^+$ sweeps, and the two orderings that witness $\LexCycle = 2$ are duals of one another. 
\end{corollary}
\begin{proof}
    As mentioned before, Theorem \ref{thm:DusartHabib} implies that in at most $n$ sweeps one gets a cocomparability order. 
    From this point we know that we can view the successive sweeps through the incidence matrix of the edges between the two cliques. 
    Each sequence of two consecutive sweeps corresponds to sorting the columns and then the rows of the matrix, and therefore by Proposition~\ref{prop:matrix} above we get that this converges in less that $2n$ more sweeps to a fixed matrix. 
    Now this means that we have reached a cycle of length $2$ and that indeed the two orderings are duals one of another.
\end{proof}

\section{Conclusion \& Perspectives}
\label{conclusion}
In this paper, we study a new graph parameter, $\LexCycle$, which measures the maximum length of a cycle of $\LexBFS^+$ sweeps. 
We believe it reflects some measure of linearity structure of a class of graphs: if the class has some strong linear structure, $\LexCycle$ should be small for this class. 
For example, interval graphs have a clear definition of linearity: they are the graphs that arise from the intersection of intervals on the real line. 
Cocomparability graphs, and more generally AT-free graphs, also have a notion of linearity. 
Every AT-free graph, and thus every cocomparability graph, admits a dominating diametral path~\cite{COS99}. 
This dominating diametral path has been known in the literature as the \emph{spine} of the graph~\cite{COS99}, because it opens the graph in a linear fashion, where vertices not on path hang at distance one.  
We have no example of an asteroidal triple-free graph whose $\LexCycle$ is larger than 2, so it may be the case that Conjecture~\ref{conj:DusartHabib} is even true for $AT$-free graphs (which we note have asteroidal number 2). 

Note however that as shown in this paper, a stronger conjecture stating $\LexCycle$ is bounded above by the asteroidal number is false (this was the question of Stacho for which we provided a counterexample in Section \ref{sec:multisweep}). 
Our construction suggests that perhaps $\LexCycle$ cannot be bounded by any polynomial function on the number of vertices.

Towards proving Conjecture~\ref{conj:DusartHabib} about cocomparability graphs, we showed that domino-free cocomparability graphs (which contain cographs, interval graphs, cobipartite graphs) all have $\LexCycle = 2$.
Next, we motivate the choice of the domino as a forbidden structure and a potential direction to prove the conjecture for cocomparability graphs. 

Define a \textbf{$k$-ladder} to be an induced graph of $k$ \emph{chained} $C_4$s. 
More precisely, a ladder is a graph $H=(V_H, E_H)$ where $V_H = \{x_0, x_1, x_2, \ldots, x_k, y_0, y_1, \ldots, y_k\}$ and $E_H =  \{ (x_i,y_i), (x_i, x_{i+1}), (y_i, y_{i+1}) : \forall i,  0\leq i \leq k-1\} \cup \{(x_k, y_k)\}$, as illustrated in Figure \ref{ladder}.
\begin{figure}[H]
	\centering
	\begin{tikzpicture}
	[scale=.3]
	\node[circle, draw, fill=black!100, inner sep=1pt, minimum width=4pt,label=above:$x_0$] (a) at (0,0) {};
	\node[circle, draw, fill=black!100, inner sep=1pt, minimum width=4pt,label=below:$y_0$] (b) at (0,-4) {};
	
	\node[circle, draw, fill=black!100, inner sep=1pt, minimum width=4pt,label=above:$x_1$] (c) at (4,0) {};
	\node[circle, draw, fill=black!100, inner sep=1pt, minimum width=4pt,label=below:$y_1$] (d) at (4,-4) {};
	
	\node[circle, draw, fill=black!100, inner sep=1pt, minimum width=4pt,label=above:$x_2$] (e) at (8,0) {};
	\node[circle, draw, fill=black!100, inner sep=1pt, minimum width=4pt,label=below:$y_2$] (f) at (8,-4) {};
	
	\node (i) at (14-4,0) {$\cdots$};
	\node (j) at (14-4,-4) {$\cdots$};

	\node[circle, draw, fill=black!100, inner sep=1pt, minimum width=4pt,label=above:$x_{i-1}$] (k) at (20-4-4,0) {};
	\node[circle, draw, fill=black!100, inner sep=1pt, minimum width=4pt,label=below:$y_{i-1}$] (l) at (20-4-4,-4) {};
	\node[circle, draw, fill=black!100, inner sep=1pt, minimum width=4pt,label=above:$x_i$] (m) at (24-4-4,0) {};
	\node[circle, draw, fill=black!100, inner sep=1pt, minimum width=4pt,label=below:$y_i$] (n) at (24-4-4,-4) {};
	\node[circle, draw, fill=black!100, inner sep=1pt, minimum width=4pt,label=above:$x_{i+1}$] (o) at (28-4-4,0) {};
	\node[circle, draw, fill=black!100, inner sep=1pt, minimum width=4pt,label=below:$y_{i+1}$] (p) at (28-4-4,-4) {};

	\node (q) at (26-4,0) {$\cdots$};
	\node (r) at (26-4,-4) {$\cdots$};
	
	\node[circle, draw, fill=black!100, inner sep=1pt, minimum width=4pt,label=above:$x_{k-1}$] (s) at (28-4,0) {};
	\node[circle, draw, fill=black!100, inner sep=1pt, minimum width=4pt,label=below:$y_{k-1}$] (t) at (28-4,-4) {};
	\node[circle, draw, fill=black!100, inner sep=1pt, minimum width=4pt,label=above:$x_k$] (u) at (32-4,0) {};
	\node[circle, draw, fill=black!100, inner sep=1pt, minimum width=4pt,label=below:$y_k$] (v) at (32-4,-4) {};
		
	\foreach \from/\to in {a/b,a/c,b/d,c/d,c/e,d/f,e/f,k/l,k/m,l/n,m/n,m/o,n/p,o/p,s/t,s/u,t/v,u/v}
	\draw (\from) -- (\to);
	\end{tikzpicture}
	\caption{A $k$-ladder.}\label{ladder}
\end{figure}

Observe that interval graphs are equivalent to $1$-ladder-free cocomparability graphs, and domino-free graphs are precisely $2$-ladder-free graphs. 
Therefore we believe that the study of $k$-ladder-free cocomparability graphs is a good strategy towards proving LexCycle=2 for cocomparability graphs. 

Recently there has been progress on other subclasses of cocomparability graphs:  in~\cite{GaoXu}, the authors showed that $\overline{P_2\cup P_3}$-free cocomparability graphs, and thus diamond-free cocomparability graphs as well as cocomparability graphs with girth 4, have  $\LexCycle = 2$. 

Outside cocomparability graphs, we wonder if the conjecture holds for split graphs as well. 
And, although we haven't been able to prove the conjecture for trees, we strongly believe that it holds for trees using BFS only -- the lack of cycles on trees implies that every LexBFS is a BFS ordering. 

\noindent\emph{A word on runtime for arbitrary cocomparability graphs}: 
Although the conjecture is still open for cocomparability graphs, experimentally we observed that the convergence often happens relatively quickly, but not always, as shown by the sequence of graphs $\{G_n\}_{n \ge 2}$ presented below. 
This graph family, experimentally, takes $O(n)$ LexBFS$^+$ sweeps before converging. 
We describe an example in the family in terms of its complement since it is easier to picture, and the LexBFS traversals of the complement are easier to parse. 
Let $G_n=(V = A \cup B, E)$ be a \emph{comparability} graph on $2n+2$ vertices, where both $A$ and $B$ are paths, i.e. $A = a_1, a_2, \ldots, a_n, B = x, y, b_1, b_2, \ldots, b_n$, and the only edges in $E$ are of the form $E = \{(a_ia_{i+1}) : i \in [n-1] \} \cup \{(xy), (yb_1)\} \cup \{(b_jb_{j+1}) : j \in [n-1]\}$. 
The initial comparability ordering is constructed by collecting the odd indexed vertices first, then the even indexed ones as follows: 
\begin{itemize}
	\item Initially we start $\tau$ with $x, a_1$.
	\item In general, if the last element in $\tau$ is $a_i$ and $i$ is odd, while $i$ is in a valid range, append $b_i, b_{i+2}, a_{i+2}$ to $\tau$ and repeat.
	\item If $n$ is even, append $b_n, a_n$ to $\tau$, otherwise append $a_{n-1}, b_{n-1}$ to $\tau$.
	\item Again while $i$ is in a valid range, we append the even indexed vertices $a_{i}, b_{i}, b_{i-2}, a_{i-2}$ to $\tau$.
	\item Append $y$ to $\tau$. 
\end{itemize}
The ordering $\tau$ as constructed is a transitive orientation of the graph, and thus is a cocomparability ordering in the complement. 
We perform a series of LexBFS$^+$ sweeps where $\sigma_1 = $ LexBFS$^+(\tau)$ in the complement, i.e. the cocomparability graph.  

Every subsequent $^+$ sweep will proceed to ``gather" the elements of $A$ close to each other, resulting in an ordering that once it moves to path $A$ remains in $A$ until all its elements have been visited. 
An intuitive way to see why this must happen is to notice in the complement, the vertices of $A$ are universal to $B$ and thus must have a strong pull. 
Experimentally, this 2-paths graph family takes $O(n)$ LexBFS$^+$ sweeps before converging. 
Figure \ref{g6} below is an example for $n = 6$. 
\begin{figure}[H]
	\begin{minipage}{.3\textwidth}
		\centering
		\begin{tikzpicture}[auto=left,every node/.style={circle,inner sep=1pt, minimum width=4pt,draw,fill=black!100}]
		\node[label=left:$a_6$] (a6) at (0,2.5) {};
		\node[label=left:$a_5$] (a5) at (0,1.5) {};
		\node[label=left:$a_4$] (a4) at (0,0.5) {};
		\node[label=left:$a_3$] (a3) at (0,-0.5) {};
		\node[label=left:$a_2$] (a2) at (0,-1.5) {};
		\node[label=left:$a_1$] (a1) at (0,-2.5) {};
		
		\node[label=left:$b_6$] (b6) at (1,2.5) {};
		\node[label=left:$b_5$] (b5) at (1,1.5) {};
		\node[label=left:$b_4$] (b4) at (1,0.5) {};
		\node[label=left:$b_3$] (b3) at (1,-0.5) {};
		\node[label=left:$b_2$] (b2) at (1,-1.5) {};
		\node[label=left:$b_1$] (b1) at (1,-2.5) {};
		\node[label=left:$y$] (y) at (1,-3.5) {};
		\node[label=left:$x$] (x) at (1,-4.5) {};
		
		\foreach \from/\to in {a1/a2, a2/a3, a3/a4, a4/a5, a5/a6, x/y, y/b1, b1/b2, b2/b3, b3/b4, b4/b5, b5/b6}
		\draw (\from) -- (\to);
		\end{tikzpicture}
	\end{minipage}
	\begin{minipage}{0pt}
		\centering
		\begin{align*}
		\tau & &= x, a_1, b_1, b_3, a_3, a_5, b_5, b_6, a_6, a_4, b_4, b_2, a_2, y\\
		\sigma_1 &= \text{LexBFS}^+(\tau) &= y, a_2, b_2, b_4, a_4, a_6, b_6, a_5, b_5, b_3, a_1, a_3, x, b_1\\
		\sigma_2 &= \text{LexBFS}^+(\sigma_1) &= b_1, x, a_1, a_3, b_3, b_5, a_5, a_6, b_6, b_4, a_4, a_2, b_2, y\\
		\sigma_3 &= \text{LexBFS}^+(\sigma_2) &= y, b_2, a_2, a_4, b_4, b_6, a_6, b_5, a_1, a_1, a_3, b_3, x, b_1\\
		\sigma_4 &= \text{LexBFS}^+(\sigma_3) &= b_1, x, b_3, a_3, a_1, a_5, b_5, b_6, a_6, a_4, a_2, b_4, b_2, y\\
		\sigma_5 &= \text{LexBFS}^+(\sigma_4) &= y, b_2, b_4, a_2, a_4, a_6, b_6, a_5, a_1, a_3, b_5, b_3, x, b_1\\
		\sigma_6 &= \text{LexBFS}^+(\sigma_5) &= b_1, x, b_3, b_5, a_3, a_1, a_5, a_6, a_4, a_2, b_6, b_4, b_2, y\\
		\sigma_7 &= \text{LexBFS}^+(\sigma_6) &= y, b_2, b_4, b_6, a_2, a_4, a_6, a_5, a_1, a_3, b_5, b_3, x, b_1\\
		\sigma_8 &= \text{LexBFS}^+(\sigma_7) &= b_1, x, b_3, b_5, a_3, a_1, a_5, a_6, a_4, a_2, b_6, b_4, b_2, y &= \sigma_6
		\end{align*}
	\end{minipage}
	\caption{$G_6$, A comparability graph; $\tau$ a cocomparability ordering of the complement of $G_6$ and a series of LexBFS$^+$ of the corresponding cocomparability graph. }\label{g6}
\end{figure}

\noindent \emph{Other Graph Searches:} 
One could raise a similar cycle question for different graph searches; in particular, \emph{Lexicographic Depth First Search} (LexDFS). 
LexDFS was introduced in \cite{CK08} and is a graph search that extends DFS is a similar way to how LexBFS extends BFS - see Algorithm \ref{lexdfsalg}.
\begin{algorithm}[H]
	\caption{LexDFS}\label{lexdfsalg}
	\begin{algorithmic}[1]
		\Input A graph $G=(V, E)$ and a start vertex $s$
		\Output An ordering $\sigma$ of $V$
		\State assign the label $\epsilon$ to all vertices, and $label(s) \leftarrow \{0\}$
		\For{$i \leftarrow 1$ to $n$}
		\State pick an unnumbered vertex $v$ with lexicographically largest label
		\State $\sigma(v) \leftarrow i$ \Comment{$v$ is assigned the number $i$}
		\ForAll{unnumbered vertex $w$ adjacent to $v$}
		\State prepend $i$ to $label(w)$
		\EndFor
		\EndFor
	\end{algorithmic}
\end{algorithm}
LexDFS has led to a number of linear time algorithms on cocomparability graphs, including maximum independent set and Hamilton path \cite{corneil2013ldfs, BigArt, kohler2014linear}. 
In fact, these recent results have shown just how powerful combining LexDFS and cocomparability orderings is. 
It is therefore natural to ask whether a sequence of LexDFS orderings on cocomparability graphs reaches a cycle with nice properties. 
Unfortunately, this is not the case as shown by the example in Figure~\ref{ldfscexp}, where $G$ is a cocomparability graph as witnessed by the following cocomparability ordering $\tau = a, c, e, f, g, d, b$, however doing a sequence of LexDFS$^+$ on $G$ cycles before we reach a cocomparability ordering, and the cycle has size four. \\
\begin{figure}
	\begin{minipage}{.3\textwidth}
		\centering
		\begin{tikzpicture}[auto=left,every node/.style={circle,inner sep=1pt, minimum width=4pt,draw,fill=black!100}]
		\node[label=above:$g$] (g) at (0,2) {};
		\node[label=left:$e$] (e) at (-1,1) {};
		\node[label=right:$f$] (f) at (1,1) {};
		\node[label=below:$c$] (c) at (-0.8,0) {};
		\node[label=below:$a$] (a) at (-1.5,0) {};
		\node[label=below:$d$] (d) at (0.8,0) {};
		\node[label=below:$b$] (b) at (1.5,0) {};
		\foreach \from/\to in {g/e, g/f, g/c, e/f, e/c, f/d, c/d, c/a, d/b}
		\draw (\from) -- (\to);
		\end{tikzpicture}
	\end{minipage}
	\begin{minipage}{0pt}
		\centering
		\begin{align*}
		\sigma_1 &= \text{LexDFS}(G) &= a, c, d, b, f, g, e\\
		\sigma_2 &= \text{LexDFS}^+(\sigma_1) &= e, g, f, d, c, a, b\\
		\sigma_3 &= \text{LexDFS}^+(\sigma_2) &= b, d, c, a, g, e, f\\
		\sigma_4 &= \text{LexDFS}^+(\sigma_3) &= f, e, g, c, d, b, a\\
		\sigma_5 &= \text{LexDFS}^+(\sigma_4) &= a, c, d, b, f, g, e &= \sigma_1
		\end{align*}
	\end{minipage}
	\caption{A sequence of LexDFS$^+$ orderings on a cocomparability graph, that cycles after 5 iterations, and none of the orderings is a cocomparability order.}\label{ldfscexp}
\end{figure}

\noindent \textbf{Acknowledgements:} The authors wish to warmly thank the anonymous referees for their very helpful comments that helped us correct flaws in our proofs. They also  thank  J\'er\'emie Dusart and Derek Corneil for helpful discussions on this subject.

\end{document}